\documentclass[preprint,12pt,authoryear]{elsarticle}

\usepackage{amsmath,amssymb,amsthm,abraces}
\usepackage{graphicx}
\usepackage{enumitem}
\usepackage{hyperref}

\newtheorem{theorem}{Theorem}

\newtheorem{corollary}[theorem]{Corollary}
\newtheorem{proposition}[theorem]{Proposition}

\newtheorem{example}{Example}

\newcommand{\Prob}{\mathbb P}
\newcommand{\Bin}{\operatorname{Bin}}
\newcommand{\ULC}{\operatorname{ULC}}
\newcommand{\myqed}{\hfill $\triangle$}

\begin{document}

\begin{frontmatter}

\title{Maximum Entropy of Sums of Independent Ternary Random Variables}

\author[uns]{Mladen Kova\v{c}evi\'c\corref{cor1}}
\ead{kmladen@uns.ac.rs}
\cortext[cor1]{Corresponding author.}

\affiliation[uns]{organization={Faculty of Technical Sciences, University of Novi Sad},
            country={Serbia}}

\begin{abstract}
The classical problem of maximizing the Shannon entropy of a sum of
independent random variables supported on a finite alphabet is considered,
and settled in the ternary case.
Namely, the following theorem is established:
if \(X_1,\ldots,X_n\) are independent random variables taking values in
\(\{0,1,2\}\), then the entropy of \(S_n=X_1+\cdots+X_n\) is maximized
when \(X_1,\ldots,X_{n-1}\) are uniform on \(\{0,2\}\) and the probability
mass function of \(X_n\) is given by \(\Prob(X_n=0) = \Prob(X_n=2) = w/2\),
\(\Prob(X_n=1) = 1-w\), where \(w = \big(1 + 2^{-H(B_n)+H(B_{n-1})}\big)^{-1}\)
and \(B_m\sim \Bin(m,1/2)\).
The statement can be seen as an extension to ternary alphabets of the
Shepp--Olkin--Mateev theorem.
The proof uses the Hermite--Biehler theorem, Newton's inequalities, and Yu's
maximum-entropy theorem for ultra-log-concave distributions.
\end{abstract}

\begin{keyword}
Maximum entropy \sep sums of random variables \sep binomial distribution \sep
Shepp--Olkin theorem \sep ultra-log-concavity
\MSC[2020] 94A17 \sep 60C05 \sep 60G50
\end{keyword}

\end{frontmatter}

\section{Introduction}
\label{sec:intro}

Entropy maximization problems occupy an important place in probability,
statistics, information theory, and statistical mechanics \citep{jaynes,cover+thomas}.
Many of the most familiar probability laws -- uniform, Gaussian, exponential,
geometric, binomial, Poisson, etc.\ -- can be characterized as the ``most random''
distributions subject to natural constraints.
Such results are useful not only as variational characterizations of classical
distributions, but also because they often reveal hidden convexity, log-concavity,
or stability phenomena behind probabilistic inequalities.

A particularly elegant problem in this context concerns the entropy of sums
of independent random variables.
Let \(X_1,\ldots,X_n\) be independent random variables taking values in
\(\{0,1,\ldots,r\}\), and set
\begin{equation}
\label{eq:sum-definition-intro}
S_n = X_1 + \cdots + X_n .
\end{equation}
For \(r=1\), this is the classical Bernoulli-sum problem.
The Shepp--Olkin--Mateev theorem \citep{shepp+olkin,mateev} asserts that \(H(S_n)\)
is maximized when all \(X_i\) are \(\operatorname{Bernoulli}(1/2)\).
Equivalently, among all sums of independent Bernoulli random variables, the
entropy-maximizing one is the symmetric binomial distribution.
This result has inspired a substantial line of work on related problems
\citep{harremoes,hillion+johnson,hillion+johnson2,hillion+johnson3,johnson,johnson3,yu,yu2}.

For non-binary alphabets, however, very little is known.
The naive guess that each \(X_i\) should be uniform on \(\{0,1,\ldots,r\}\) is
false.
In fact, the entropy-maximizing configuration appears to be highly nonuniform.
The present author conjectured in \citep{kovacevic} that the maximum is attained
by placing \(n-1\) of the variables on the two endpoints \(0\) and \(r\), each
uniformly, while the remaining variable is a mixture of the uniform distribution
on the endpoints \(\{0,r\}\) and the uniform distribution on the interior points
\(\{1,\ldots,r-1\}\).
More precisely, the conjectured maximum is
\begin{equation}
\label{eq:conjectured-maximum}
w_0H(B_n)+(1-w_0)\bigl(H(B_{n-1})+\log_2(r-1)\bigr)+h(w_0),
\end{equation}
where \(B_m\sim\operatorname{Bin}(m,1/2)\) and
\begin{equation}
\label{eq:conjectured-weight}
w_0 = \frac{1}{1 + (r-1) 2^{-H(B_n)+H(B_{n-1})}} .
\end{equation}
This formula is consistent with the binary Shepp--Olkin--Mateev theorem, and
was proved in \citep{kovacevic} in several other special cases, including \(n=2\)
for arbitrary \(r\).
In this note we prove the conjecture for the first non-binary alphabet, that is,
for \(r=2\) and arbitrary \(n\).

The conjecture was inspired in part by the continuous version of the problem
\citep{ordentlich,yu2} and suggests an interesting qualitative phenomenon.
In the entropy-maximizing configuration, most variables do not use the whole
alphabet -- they use only the two extreme symbols.
The remaining variable is responsible for distributing mass among the residue
classes modulo \(r\).
Thus the conjecture separates the problem into two effects: binomial spreading
along each residue class, and entropy of the residue class itself.
Unfortunately, the proof method based on this observation does not seem to extend
to \(r \geqslant 3\).

The note is organized as follows.
Section~\ref{sec:prelim} recalls the needed facts about ultra-log-concavity,
the Hermite--Biehler theorem, Newton's inequalities, and Yu's maximum-entropy
theorem.
Section~\ref{sec:parity} proves the ultra-log-concavity of the parity-conditioned
laws, i.e., of the distributions in each of the two residue classes.
In Section~\ref{sec:main} we state and prove our main result -- the maximum
entropy theorem for sums of independent ternary random variables.
Section~\ref{sec:example} presents a counterexample showing that our residue-class
ultra-log-concavity approach fails for larger alphabets.

\section{Preliminaries}
\label{sec:prelim}

The Shannon entropy of a discrete random variable \(X\) taking values in
\(\{0, \ldots, r\}\) is defined as
\begin{equation}
\label{eq:entropy-definition}
H(X) = - \sum_{j=0}^r \Prob(X=j) \log_2 \Prob(X=j) .
\end{equation}
In the binary case we also write \begin{equation}
\label{eq:binary-entropy}
h(p) = -p \log_2 p - (1-p) \log_2(1-p).
\end{equation}

A nonnegative sequence \((u_0,\ldots,u_m)\) is said to be log-concave if the
following inequalities hold
\begin{equation}
\label{eq:log-concavity}
u_k^2 \geqslant u_{k-1}u_{k+1},  \qquad  k =  1, \ldots, m-1 .
\end{equation}
It is said to be ultra-log-concave of order \(m\), abbreviated \(\ULC(m)\), if
the sequence
\begin{equation}
\label{eq:ulc-normalized}
\bigg(\frac{u_k}{\binom{m}{k}}\bigg)_{\! k=0}^{\! m}
\end{equation}
is log-concave, or, equivalently, if
\begin{equation}
\label{eq:ulc-inequality}
k(m-k)u_k^2  \geqslant  (k+1)(m-k+1)u_{k-1}u_{k+1},  \qquad  k = 1, \ldots, m-1 .
\end{equation}

We shall need in the sequel the following maximum-entropy theorem of Yu \citep{yu}.

\begin{theorem}[Yu]
\label{thm:Yu}
Let \(Y\) be a random variable taking values in \(\{0,\ldots,m\}\).
If the probability mass function of \(Y\) is \(\ULC(m)\), then
\begin{equation}
\label{eq:yu-bound}
H(Y) \leqslant H(B_m) ,
\end{equation}
where \(B_m \sim \Bin(m,1/2)\).
\end{theorem}

We shall also need a particular consequence of the Hermite--Biehler theorem
\citep[Theorem 6.3.4]{rahman+schmeisser} stated below.

\begin{theorem}[Hermite--Biehler]
\label{thm:HBorig}
Let \(R\) and \(I\) be non-constant polynomials with real coefficients.
Then \(R\) and \(I\) have strictly interlacing zeros if and only if the polynomial
\(F(z) = R(z) + i I(z)\) has all its zeros either in the open upper half-plane or
in the open lower half-plane.
\end{theorem}

\begin{corollary}
\label{thm:HB}
Let \(P\) be a real polynomial all of whose zeros lie in the open left half-plane.
Write
\begin{equation}
\label{eq:hb-decomposition}
P(z) = E(z^2) + z O(z^2) .
\end{equation}
Then \(E\) and \(O\) have only real nonpositive zeros.
\end{corollary}
\begin{proof}
Set \(F(x) = P(ix)\).
Since \(P\) has all zeros in the open left half-plane, \(F\) has all zeros in the
open upper half-plane.
Moreover,
\begin{equation}
\label{eq:hb-proof-transform}
F(x) = E(-x^2) + i x O(-x^2) .
\end{equation}
By Theorem \ref{thm:HBorig}, the real and imaginary parts of \(F\), namely
\(E(-x^2)\) and \(xO(-x^2)\), have only real zeros.
It follows immediately that all zeros of \(E\) and \(O\) are real and nonpositive.
\end{proof}

Finally, we recall Newton's inequalities in the following form \citep[Theorem 2]{stanley}.

\begin{theorem}[Newton]
\label{thm:Newton}
If a polynomial \(F(t)=\sum_{k=0}^m f_k t^k\) with real coefficients has only
real zeros, then
\begin{equation}
\label{eq:newton-inequalities}
k(m-k)f_k^2  \geqslant  (k+1)(m-k+1)f_{k-1}f_{k+1},  \qquad  k = 1, \ldots, m-1 .
\end{equation}
In particular, if the coefficients \(f_k\) are nonnegative, the sequence
\((f_0, \ldots, f_m)\) is \(\ULC(m)\).
\end{theorem}

\section{Parity-conditioned laws}
\label{sec:parity}

Let \(X_1,\ldots,X_n\) be independent random variables taking values in
\(\{0,1,2\}\).
Denote
\begin{equation}
\label{eq:ternary-probabilities}
\Prob(X_i=0) = a_i , \qquad  \Prob(X_i=1) = b_i , \qquad  \Prob(X_i=2) = c_i .
\end{equation}
Let
\begin{equation}
\label{eq:sum-definition-parity}
S_n = X_1 + \cdots + X_n .
\end{equation}
Define
\begin{equation}
\label{eq:parity-variables}
J = S_n \bmod 2 , \qquad   K = \frac{S_n-J}{2} .
\end{equation}
Thus \(J\in\{0,1\}\), and \(S_n\) is determined by the pair \((J,K)\).

Consider the generating polynomial
\begin{equation}
\label{eq:generating-polynomial}
P(z) = \prod_{i=1}^n (a_i+b_i z+c_i z^2) .
\end{equation}
Write its even--odd decomposition as
\begin{equation}
\label{eq:even-odd-decomposition}
P(z) = E(z^2) + zO(z^2) ,
\end{equation}
where
\begin{equation}
\label{eq:even-polynomial}
E(t) = \sum_{k=0}^n e_k t^k ,  \qquad  O(t) = \sum_{k=0}^{n-1} o_k t^k .
\end{equation}
Then
\begin{equation}
\label{eq:odd-polynomial}
e_k = \Prob(S_n=2k) ,  \qquad  o_k = \Prob(S_n=2k+1) .
\end{equation}
Consequently, conditional on \(J=0\), \(K\) has probability mass function
proportional to \((e_0,\ldots,e_n)\), and conditional on \(J=1\), \(K\) has
probability mass function proportional to \((o_0,\ldots,o_{n-1})\).

\begin{proposition}
\label{prop:parity-ulc}
The sequence \((e_0, \ldots, e_n)\) is \(\ULC(n)\), and the sequence\linebreak
\((o_0, \ldots, o_{n-1})\) is \(\ULC(n-1)\).
\end{proposition}

\begin{proof}
We first assume that \(a_i, b_i, c_i > 0\) for all \(i\).
For each \(i\), the quadratic polynomial
\begin{equation}
\label{eq:even-odd-coefficients}
a_i+b_i z+c_i z^2
\end{equation}
has all zeros in the open left half-plane.
Indeed, if its zeros are real, then their product is \(a_i/c_i>0\) and their
sum is \(-b_i/c_i<0\), so both zeros are negative.
If its zeros are not real, then they form a conjugate pair with real part
\(-b_i/(2c_i) < 0\).
Therefore \(P\) has all zeros in the open left half-plane.
By Corollary \ref{thm:HB}, applied to
\begin{equation}
\label{eq:quadratic-factor}
P(z) = E(z^2) + z O(z^2) ,
\end{equation}
both \(E\) and \(O\) have only real (and nonpositive) zeros.
By Theorem \ref{thm:Newton}, the sequences of coefficients of these polynomials
are ULC, as claimed.

It remains to remove the positivity assumption.
If some coefficients vanish, replace each triple \((a_i, b_i, c_i)\) by
\begin{equation}
\label{eq:epsilon-regularization}
\left(\frac{a_i+\varepsilon}{1+3\varepsilon},
         \frac{b_i+\varepsilon}{1+3\varepsilon},
         \frac{c_i+\varepsilon}{1+3\varepsilon}\right) ,
   \qquad \varepsilon > 0 .
\end{equation}
The corresponding even and odd coefficient sequences are \(\ULC(n)\) and\linebreak
\(\ULC(n-1)\), respectively.
Letting \(\varepsilon\downarrow0\), and using the closedness of the
ultra-log-concavity inequalities, gives the claim.
\end{proof}

\begin{corollary}
\label{thm:corbin}
For every choice of independent ternary random variables \(X_1, \ldots, X_n\),
\begin{equation}
\label{eq:cor-even-entropy}
H(K\mid J=0) \leqslant H(B_n) ,
\end{equation}
and
\begin{equation}
\label{eq:cor-odd-entropy}
H(K\mid J=1)  \leqslant  H(B_{n-1}) .
\end{equation}
Consequently,
\begin{equation}
\label{eq:cor-conditional-entropy}
H(K\mid J)  \leqslant  \Prob(J=0)H(B_n) + \Prob(J=1)H(B_{n-1}) .
\end{equation}
\end{corollary}

\begin{proof}
If \(\Prob(J=0) > 0\), the conditional law of \(K\) given \(J=0\) is obtained by
normalizing the sequence \((e_0, \ldots, e_n)\).
Normalization preserves ultra-log-concavity, so this conditional law is \(\ULC(n)\).
By Theorem \ref{thm:Yu} we then have
\begin{equation}
\label{eq:sum-definition-main}
H(K\mid J=0)  \leqslant  H(B_n) .
\end{equation}
If \(\Prob(J=0)=0\), the term \(\Prob(J=0)H(K\mid J=0)\) is irrelevant.

Similarly, if \(\Prob(J=1) > 0\), the conditional law of \(K\) given \(J=1\) is
obtained by normalizing \((o_0, \ldots, o_{n-1})\) and is therefore \(\ULC(n-1)\).
By Theorem \ref{thm:Yu} we then have
\begin{equation}
\label{eq:main-upper-bound}
H(K\mid J=1)  \leqslant  H(B_{n-1}).
\end{equation}
Averaging over \(J\) yields \eqref{eq:cor-conditional-entropy}.
\end{proof}

\section{The maximum entropy theorem}
\label{sec:main}

\begin{theorem}
\label{thm:main}
Let \(X_1, \ldots, X_n\) be independent random variables taking values in
\(\{0,1,2\}\), and let
\begin{equation}
\label{eq:main-weight}
S_n = X_1 + \cdots + X_n .
\end{equation}
Then
\begin{equation}
\label{eq:maximizing-distribution}
H(S_n)
   \leqslant  w_0 H(B_n) + (1-w_0) H(B_{n-1}) + h(w_0) ,
\end{equation}
where
\begin{equation}
\label{eq:main-parity-variables}
w_0 = \frac{1}{1 + 2^{-H(B_n)+H(B_{n-1})}} .
\end{equation}
The equality is attained by taking \(X_1, \ldots, X_{n-1}\) uniform on \(\{0, 2\}\),
and taking \(X_n\) with distribution
\begin{equation}
\label{eq:entropy-chain-rule}
\Prob(X_n=0) = \Prob(X_n=2) = \frac{w_0}{2} ,  \qquad  \Prob(X_n=1) = 1 - w_0 .
\end{equation}
\end{theorem}

\begin{proof}
Let
\begin{equation}
\label{eq:main-parity-probability}
J = S_n \bmod 2,  \qquad  K = \frac{S_n-J}{2} .
\end{equation}
Since \(S_n = 2K + J\), the pair \((K, J)\) determines \(S_n\), and hence
\begin{equation}
\label{eq:binary-entropy-main}
H(S_n) = H(K,J) = H(J) + H(K\mid J) .
\end{equation}
Let
\begin{equation}
\label{eq:main-conditional-bound}
w = \Prob(J=0) .
\end{equation}
Then
\begin{equation}
\label{eq:main-bound-w}
H(J) = h(w) .
\end{equation}
By Corollary \ref{thm:corbin} we have
\begin{equation}
\label{eq:attaining-sum-even}
H(K\mid J)  \leqslant  w H(B_n) + (1-w) H(B_{n-1}) ,
\end{equation}
and therefore
\begin{equation}
\label{eq:attaining-even-law}
H(S_n)  \leqslant  h(w) + w H(B_n) + (1-w) H(B_{n-1}) .
\end{equation}
It remains to maximize the right-hand side over \(w \in [0,1]\).
Since this expression is strictly concave in \(w\), it is easily found by
differentiation that the maximizer is
\begin{equation}
\label{eq:attaining-sum-odd}
w_0 = \frac{1}{1 + 2^{-H(B_n)+H(B_{n-1})}} .
\end{equation}

It is left to check that the upper bound in \eqref{eq:maximizing-distribution} can be achieved.
This was already done in \citep[Theorem 2.1]{kovacevic}, but we repeat the
construction for completeness.
Let \(X_1, \ldots, X_{n-1}\) be independent and uniform on \(\{0,2\}\), and
let \(X_n\) be independent of them with
\begin{equation}
\label{eq:attaining-odd-law}
\Prob(X_n=0) = \Prob(X_n=2) = \frac{w_0}{2} ,  \qquad  \Prob(X_n=1) = 1 - w_0 .
\end{equation}
If \(X_n \in \{0,2\}\), then
\begin{equation}
\label{eq:extension-generating-polynomial}
S_n = 2Y ,
\end{equation}
where \(Y \sim B_n\).
This event has probability \(w_0\), and corresponds to \(J = 0\).
Thus
\begin{equation}
\label{eq:extension-even-odd-decomposition}
K \mid J = 0 \sim B_n .
\end{equation}
If \(X_n = 1\), then
\begin{equation}
\label{eq:example-base-polynomial}
S_n = 2Y + 1 ,
\end{equation}
where \(Y \sim B_{n-1}\).
This event has probability \(1 - w_0\), and corresponds to \(J = 1\).
Thus
\begin{equation}
\label{eq:example-cube}
K \mid J = 1 \sim B_{n-1} .
\end{equation}
Therefore equality holds in the conditional entropy bound, and also \(H(J)=h(w_0)\).
Hence the right-hand side of \eqref{eq:maximizing-distribution} is attained.
\end{proof}

Figure \ref{fig:distribution} illustrates the conditional distributions \(K \mid J\)
in the entropy-maximizing case.

\begin{figure}
\centering
  \includegraphics[width=0.7\columnwidth]{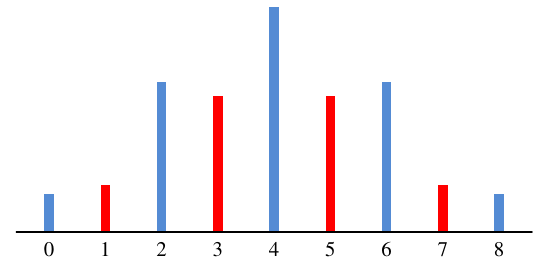}
\caption{The probability distribution of \(S_4\) in the case when the
distributions of \(X_i\) are chosen according to Theorem \ref{thm:main}.
For the purpose of illustration, the probability masses assigned to even,
resp.\ odd, values (corresponding to the case \(J=0\), resp.\ \(J=1\)) are
marked in blue, resp.\ red.
Note that \(K \mid J = 0 \sim B_4\) and \(K \mid J = 1 \sim B_3\).}
\label{fig:distribution}
\end{figure}%

\section{On extensions to larger alphabets}
\label{sec:example}

To conclude the paper, we note that, unfortunately, the proof method used to
establish Theorem \ref{thm:main} does not carry over to larger alphabets.

Recall that the main point in the proof was to control the two conditional
laws of \(K\), according to whether \(S_n\) is even or odd.
Namely, the probability generating function of \(S_n\),
\begin{equation}
\label{eq:example-residue-decomposition}
P(z) = \prod_{i=1}^n (a_i + b_i z + c_i z^2) ,
\end{equation}
is first decomposed into even and odd parts as
\begin{equation}
\label{eq:display-45}
P(z) = E(z^2) + z O(z^2) .
\end{equation}
The coefficients of \(E\) and \(O\) are precisely the (unnormalized) laws of
\(K\) conditioned on the two parity classes.
Since each quadratic \(a_i + b_i z + c_i z^2\) has all zeros in the left
half-plane (because \(a_i,b_i,c_i \geqslant 0\)), so does the product \(P\).
The Hermite--Biehler theorem then implies that \(E\) and \(O\) have only real
zeros, and Newton's inequalities imply that their coefficient sequences are
ultra-log-concave of the appropriate orders.
Finally, Yu's maximum-entropy theorem for ultra-log-concave distributions
bounds the two conditional entropies by \(H(B_n)\) and \(H(B_{n-1})\),
respectively.

The above argument illustrates both why the ternary case is tractable and why
the general conjecture remains difficult.
For alphabets of size at least four, the analogous one-variable probability
generating functions have degree \(r \geqslant 3\), and positivity of their
coefficients no longer implies that their roots are all in the left half-plane.
Indeed, as the example below shows, the residue-class polynomials need not
be real-rooted, and the sequences of their coefficients need not be ultra-log-concave.
This is true already for \(r=3\) and \(n=3\).
Thus the method used in the present paper does not simply iterate to larger
alphabets; proving the full conjecture, if it is indeed true, will likely
require a different argument.

\begin{example}
\label{ex:larger-alphabet-obstruction}
\textnormal{
Consider the case \(r=3\), \(n=3\).
Let
\begin{equation}
\label{eq:display-46}
p(z) = 0.15 + 0.06z + 0.70z^2 + 0.09z^3
\end{equation}
and
\begin{equation}
\label{eq:display-47}
P(z) = p(z)^3 .
\end{equation}
Write
\begin{equation}
\label{eq:display-48}
P(z) = P_0(z^3) + z P_1(z^3) + z^2 P_2(z^3) .
\end{equation}
A direct calculation gives
\begin{equation}
\label{eq:example-residue-polynomials}
\begin{aligned}
   P_0(t) &= 0.003375 + 0.044091 t + 0.369325 t^2 + 0.000729 t^3 ,  \\
   P_1(t) &= 0.00405 + 0.23292 t + 0.133758 t^2 ,  \\
   P_2(t) &= 0.04887 + 0.145872 t + 0.01701 t^2 .
\end{aligned}
\end{equation}
The polynomial \(P_0\) is not real-rooted -- it has one real zero and one
complex conjugate pair.
Moreover, the coefficients of \(P_0\) are not ultra-log-concave of order \(3\)
because the inequality \(u_1^2 \geqslant 3 u_0 u_2\) does not hold (here \(u_k\)
are the coefficients of \(P_0\)).
This means that the conditional distribution corresponding to the residue class
\(S_n \equiv 0 \pmod r\) is not \(\ULC(n)\), and therefore Theorem \ref{thm:Yu}
cannot be used to upper-bound its entropy.}
\myqed
\end{example}

\section*{Funding}\label{sec:funding}
This work was supported by the Ministry of Science, Technological Development
and Innovation of the Republic of Serbia (contract no. 451-03-34/2026-03/200156)
and by the Faculty of Technical Sciences, University of Novi Sad, Serbia (project
no. 01-3609/1).


\section*{Declaration of generative AI and AI-assisted technologies in the manuscript preparation process}\label{sec:generative-ai}
During the preparation of this work, the author used ChatGPT (OpenAI) to verify some of the derivations and assist with
formatting the manuscript. The author
reviewed and edited the content as needed and takes full responsibility for the
content of the published article.


\begin{thebibliography}{99}

\bibitem[Cover and Thomas(2006)]{cover+thomas}
   T. M. Cover and J. A. Thomas,
   \emph{Elements of Information Theory},
   2nd ed., John Wiley and Sons, Inc., 2006.
\bibitem[Harremo\"es(2001)]{harremoes}
  P. Harremo\"es,
	``Binomial and {P}oisson distributions as maximum entropy distributions,''
	\emph{IEEE Trans. Inform. Theory} {47}(5) (2001), 2039--2041.
	https://doi.org/10.1109/18.930936
\bibitem[Hillion and Johnson(2016)]{hillion+johnson}
  E. Hillion and O. T. Johnson,
	``Discrete versions of the transport equation and the {S}hepp--{O}lkin conjecture,''
  \emph{Ann. Probab.} {44}(1) (2016), 276--306.
	https://doi.org/10.1214/14-AOP973
\bibitem[Hillion and Johnson(2017)]{hillion+johnson2}
  E. Hillion and O. T. Johnson,
	``A proof of the {S}hepp--{O}lkin entropy concavity conjecture,''
  \emph{Bernoulli} {23}(4B) (2017), 3638--3649.
	https://doi.org/10.3150/16-BEJ860
\bibitem[Hillion and Johnson(2019)]{hillion+johnson3}
  E. Hillion and O. T. Johnson,
	``A proof of the {S}hepp--{O}lkin entropy monotonicity conjecture,''
  \emph{Electron. J. Probab.} {24} (2019), article no.~126.
	https://doi.org/10.1214/19-EJP380
\bibitem[Jaynes(1957)]{jaynes}
  E. T. Jaynes,
	``Information theory and statistical mechanics,''
	\emph{Phys. Rev.} 106(4) (1957), 620--630. 
	https://doi.org/10.1103/PhysRev.106.620
\bibitem[Johnson(2007)]{johnson}
  O. T. Johnson,
  ``Log-concavity and the maximum entropy property of the {P}oisson distribution,''
  \emph{Stochastic Process. Appl.} {117}(6) (2007), 791--802.
	https://doi.org/10.1016/j.spa.2006.10.006
\bibitem[Johnson et al.(2013)]{johnson3}
  O. Johnson, I. Kontoyiannis, and M. Madiman,
  ``Log-concavity, ultra-log-concavity, and a maximum entropy property of discrete compound Poisson measures,''
	\emph{Discrete Applied Math.} {161}(9) (2013), 1232--1250.
	https://doi.org/10.1016/j.dam.2011.08.025
\bibitem[Kova\v cevi\'c(2021)]{kovacevic}
  M. Kova\v cevi\'c,
  ``On the maximum entropy of a sum of independent discrete random variables,''
  \emph{Theory Probab. Appl.} {66}(3) (2021), 482--487.
  https://doi.org/10.1137/S0040585X97T99054X
\bibitem[Mateev(1978)]{mateev}
  P. Mateev,
	``The entropy of the multinomial distribution,''
	\emph{Theory Probab. Appl.} {23}(1) (1978), 188--190.
	https://doi.org/10.1137/1123020
\bibitem[Ordentlich(2006)]{ordentlich}
  E. Ordentlich,
	``Maximizing the entropy of a sum of independent bounded random variables,''
  \emph{IEEE Trans. Inform. Theory} {52}(5) (2006), 2176--2181.
	https://doi.org/10.1109/TIT.2006.872858
\bibitem[Rahman and Schmeisser(2002)]{rahman+schmeisser}
  Q. I. Rahman and G. Schmeisser,
  \emph{Analytic Theory of Polynomials},
  Oxford University Press, 2002.
\bibitem[Shepp and Olkin(1981)]{shepp+olkin}
  L. A. Shepp and I. Olkin,
	``Entropy of the sum of independent {B}ernoulli random variables and of the multinomial distribution,''
	Tech. Report 131, Stanford University, 1978.
	Reprinted in \emph{Contributions to probability}, Academic Press, New York, 1981, pp.~201--206.
	https://doi.org/10.1016/B978-0-12-274460-0.50022-9
\bibitem[Stanley(1989)]{stanley}
  R. P. Stanley,
	``Log-concave and unimodal sequences in algebra, combinatorics, and geometry,''
  \emph{Ann. New York Acad. Sci.} {576} (1989), 500--535.
	https://doi.org/10.1111/j.1749-6632.1989.tb16434.x
\bibitem[Yu(2008a)]{yu2}
  Y. Yu,
	``Maximum entropy for sums of symmetric and bounded random variables: a short derivation,''
	\emph{IEEE Trans. Inform. Theory} {54}(4) (2008), 1818--1819.
	https://doi.org/10.1109/TIT.2008.917660
\bibitem[Yu(2008b)]{yu}
  Y. Yu,
	``On the maximum entropy properties of the binomial distribution,''
	\emph{IEEE Trans. Inform. Theory} {54}(7) (2008), 3351--3353.
	https://doi.org/10.1109/TIT.2008.924715

\end{thebibliography}
\end{document}